%% file: ms.tex
\documentclass[11pt]{article}

\usepackage[utf8]{inputenc}

\usepackage{amsmath,amsfonts,amsthm,amssymb,color}
\usepackage[usenames,dvipsnames,svgnames,table]{xcolor}
\definecolor{darkgreen}{rgb}{0.0,0,0.9}
\usepackage{mathtools}
\usepackage{authblk}
\usepackage{fullpage}
\usepackage{parskip}
\usepackage{comment}
\usepackage{tikz}
\usepackage{bbm}
\usepackage{dsfont}
\usepackage[sc]{mathpazo}
\usepackage[basic]{complexity}
\usepackage{algorithm2e}
\usepackage[colorlinks=true,
citecolor=OliveGreen,linkcolor=BrickRed,urlcolor=BrickRed,pdfstartview=FitH]{hyperref}
\usepackage[capitalize,nameinlink]{cleveref}
\usepackage{tcolorbox}
\usepackage[short]{optidef}

\newtcolorbox{wbox}
{
	colback  = white,
}

\SetKwInOut{Input}{Input}
\SetKwInOut{Output}{Output}
\SetKwFunction{Uncross}{\textsc{Uncross}}
\SetKwFunction{MergeUncross}{\textsc{MergeUncross}}
\SetKwFunction{PerfectMatching}{\textsc{PerfectMatching}}
\SetKwBlock{InParallel}{in parallel do}{end}
\SetKwFor{ParallelFor}{for}{in parallel do}{end}

\let\R\relax
\newcommand*{\R}{\mathbb{R}}

\newcommand*{\suppress}[1]{}

\newcommand*{\cI}{\mathcal{I}}
\newcommand*{\cM}{\mathcal{M}}

\newcommand*{\cR}{\mathcal{R}}

\newcommand*{\oG}{\overline{G}}
\newcommand*{\oE}{\overline{E}}
\newcommand{\sat}{\mbox{\rm satisfaction}}
\newcommand{\worth}{\mbox{\rm worth}}

\makeatletter
\def\thm@space@setup{%
	\thm@preskip= 10pt
	\thm@postskip=\thm@preskip 
}
\makeatother

\makeatletter
\renewcommand{\paragraph}{%
	\@startsection{paragraph}{4}%
	{\z@}{5pt}{-1em}%
	{\normalfont\normalsize\bfseries}%
}
\makeatother

\newtheorem{theorem}{Theorem}

\newtheorem{corollary}{Corollary}

\theoremstyle{definition}
\newtheorem{definition}{Definition}

\newtheorem{remark}{Remark}

\newenvironment{fminipage}%
{\begin{Sbox}\begin{minipage}}%
		{\end{minipage}\end{Sbox}\fbox{\TheSbox}}

\newcommand\QQ{\boldsymbol{\mathit{Q}}}

\newcommand\RR{\boldsymbol{\mathit{R}}}

\title{Cores of Games via Total Dual Integrality, \\
with Applications to Perfect Graphs and Polymatroids}

\author[1]{Vijay V.~Vazirani\footnote{Supported in part by NSF grants CCF-1815901 and CCF-2230414.}}

\affil[1]{University of California, Irvine}

\date{}

\begin{document}
	\maketitle
	
\begin{abstract}
LP-duality theory has played a central role in the study of cores of games, right from the early days of this notion to the present time. The classic paper of Shapley and Shubik \cite{Shapley1971assignment} introduced the ``right'' way of exploiting the power of this theory, namely picking problems whose LP-relaxations support polyhedra having integral vertices. So far, the latter fact was established by showing that the constraint matrix of the underlying linear system is {\em totally unimodular}.
	
We attempt to take this methodology to its logical next step -- {\em using total dual integrality} -- thereby addressing new classes of games which have their origins in two major theories within combinatorial optimization, namely perfect graphs and polymatroids. In the former, we address the stable set and clique games and in the latter, we address the matroid independent set game. For each of these games, we prove that the set of core imputations is precisely the set of optimal solutions to the dual LPs.   

Another novelty is the way the worth of the game is allocated among sub-coalitions. Previous works follow the {\em bottom-up process} of allocating to individual agents; the  allocation to a sub-coalition is simply the sum of the allocations to its agents. The {\em natural process for our games is top-down}. The optimal dual allocates to ``objects'' in the grand coalition; a sub-coalition inherits the allocation of each object with which it has non-empty intersection.  
\end{abstract}

\bigskip
\bigskip
\bigskip
\bigskip
\bigskip
\bigskip
\bigskip
\bigskip
\bigskip
\bigskip
\bigskip
\bigskip
\bigskip
\bigskip
\bigskip
\bigskip
\bigskip
\bigskip

\pagebreak

\input{Intro}

\input{Related}

\input{Vertices}

\input{Perfect}

\input{Polymatroid}

\input{ack}

	\bibliographystyle{alpha}
	\bibliography{refs}

\end{document}

%% file: Intro.tex
\section{Introduction}
\label{sec.intro}

The {\em core} is a  quintessential solution concept in cooperative game theory. It captures all possible ways of distributing the total worth of a game among individual agents and sub-coalitions in such a way that the grand coalition remains intact, i.e., a sub-coalition will not be able to generate more profits by itself and therefore has no incentive to secede from the grand coalition. 

LP-duality theory has played a central role in the study of the core, right from its early days all the way to the present time. The early works of Bondareva \cite{Bondareva1963some} and Shapley \cite{Shapley1965balanced} gave a necessary and sufficient condition for non-emptiness of the core of a game, namely that it is ``balanced''. The proof of this fact is based on LP-duality theory and uses the Farkas Lemma. Although this is an insightful fact, the LP employed for proving it involves exponentially many variables, one for each subset of players. Consequently, to the best of our knowledge, it did not lead to an efficient way of computing core imputations for any natural game. 

The classic paper of Shapley and Shubik \cite{Shapley1971assignment} introduced the ``right'' way of exploiting the power of LP-duality theory for characterizing core imputations. They did this in the context of the assignment game\footnote{The assignment game can also be viewed as a matching market in which utilities of the agents are stated in monetary terms and side payments are allowed, i.e., it is a {\em transferable utility (TU) market}. For an extensive coverage of these notions, see the book by Moulin \cite{Moulin2014cooperative}.}, for which they showed that the set of core imputations is precisely the set of optimal solutions to the dual of the LP-relaxation of the maximum weight matching problem in the underlying bipartite graph. See Section \ref{sec.Complementarity} for details. 

It is important to note that this dual LP is small --- it has one variable for each agent and one constraint for each ``team'', consisting of one agent from each side of the bipartition. In general, an optimal solution to such an LP will be fractional, i.e., it will pick agents fractionally. Clearly, that is not very meaningful. In retrospect, {\em a key ingenuity of Shapley and Shubik} lay in choosing a game whose LP-relaxation has the following important property: {\em the polytope defined by its constraints has integral vertices}; in this case, they are matchings in the graph. 

Deng et al. \cite{Deng1999algorithms} distilled the ideas underlying the Shapley-Shubik Theorem to obtain a general framework, see Section \ref{sec.Deng}, which helps characterize the cores of several games that are based on fundamental combinatorial optimization problems, including maximum flow in unit capacity networks both directed and undirected, maximum number of edge-disjoint $s$-$t$ paths, maximum number of vertex-disjoint $s$-$t$ paths, maximum number of disjoint arborescences rooted at a vertex $r$ and concurrent games (i.e., general graph matching games having a non-empty core). This was followed by other works in the same vein, see Section \ref{sec.related}. Similar to the Shapley-Shubik result, a central fact underlying these successes is the integrality of the vertices of their polyhedra, which was {\em so far established by showing that the underlying linear system is totally unimodular,} see Definition \ref{def.totally-unimodular}.   

We attempt to take this methodology to its logical next step -- {\em using total dual integrality (TDI)}\footnote{See Definition \ref{def.TDI}.} -- thereby addressing new classes of games which have their origins in two major theories within combinatorial optimization, namely perfect graphs and polymatroids. In the former, we address the stable set and clique games and in the latter, we address the matroid independent set game, given via a rank oracle. For each of these games, we prove that the set of core imputations is precisely the set of optimal solutions to the dual LPs.  The TDI of the underlying linear systems, and the consequent integrality of their  polytopes, plays a key role in the proofs. Our results for the perfect graph and polymatroid games are presented in Sections \ref{sec.perfect} and \ref{sec.polymatroid}, respectively.  

Perhaps the most intriguing aspect of our work is the novel manner in which we need to allocate the worth of the game among sub-coalitions. As stated in Remark \ref{rem.top-down}, whereas in the assignment game this is done via a bottom-up process, we need to follow a top-down process for both games we study. To the best of our knowledge, all previous works characterizing the cores of games follow a bottom-up process, similar to the assignment game.

\begin{definition}
	\label{def.totally-unimodular}
	Let $Ax \leq b$ be a linear system where $A$ is an $m \times n$ matrix and $b$ an $m$-dimensional vector, both with integral entries.  $A$ is said to be {\em totally unimodular (TUM)} if every submatrix of $A$ has determinant $0, 1$ or $-1$. If so, the vertices of the polytope of this linear system are integral. 
\end{definition}

\begin{definition}
	\label{def.TDI}
	Let $Ax \leq b$ be a linear system where $A$ is an $m \times n$ matrix and $b$ an $m$-dimensional vector, both with rational entries. We will say that this linear system is {\em totally dual integral (TDI)} if for any integer-valued vector $c^T$ such that the linear program
	
\[ \max \{ cx: Ax \leq b \} \]

has an optimum solution, the corresponding dual linear program has an {\em integer} optimal solution. If so, the vertices of the polytope of this linear system are integral. 
\end{definition}

Note that TDI is a more general condition than TUM for integrality of polyhedra. If $A$ is TUM then the polyhedron of the linear system $Ax \leq b$ has integral vertices for every integral vector $b$. However, even if $A$ is not TUM, for specific choices of an integral vector $b$, the vertices of the polytope of the linear system $Ax \leq b$ may be integral, and TDI may apply in this situation. It is important to remark that TDI is not a property of the polytope but of the particular linear system chosen to define it. See \cite{GLS, Sch-book} for further details.

%% file: Related.tex
\section{Related Works}
\label{sec.related}

An imputation in the core has to ensure that {\em each} of the exponentially many sub-coalitions is ``happy'' --- clearly, that is a lot of constraints. As a result, the core is non-empty only for  a handful of games, some of which are mentioned in the Introduction. A different kind of game, in which preferences are ordinal, is based on the stable matching problem defined by Gale and Shapley \cite{GaleS}. The only coalitions that matter in this game are ones formed by one agent from each side of the bipartition. A stable matching ensures that no such coalition has the incentive to secede and the set of such matchings constitute the core of this game. Vande Vate \cite{Vate1989linear} and Rothblum \cite{Rothblum1992characterization} gave linear programming formulations for stable matchings; the vertices of their underlying polytopes are integral and are stable matchings. More recently, Kiraly and Pap \cite{TDI-Kiraly2008total} showed that the linear system of Rothblum is in fact TDI. 

To deal with games having an empty core, e.g., the general graph matching game, the following two notions have been given in the past. The first is that of {\em least core}, defined by Mascher et al. \cite{Leastcore-Maschler1979geometric}. If the core is empty, there will necessarily be sets $S \subseteq V$ such that $v(S) < p(S)$ for any imputation $v$. The least core maximizes the minimum of $v(S) - p(S)$ over all sets $S \subseteq V$, subject to $v(\emptyset) = 0$ and $v(V) = p(V)$. This involves solving an LP with exponentially many constraints, though, if a separation oracle can be implemented in polynomial time, then the ellipsoid algorithm will accomplish this in polynomial time \cite{GLS}; see below for a resolution for the case of the matching game.  

A more well known notion is that of {\em nucleolus} which is contained in the least core.  After maximizing the minimum of $v(S) - p(S)$ over all sets $S \subseteq V$, it does the same for all remaining sets and so on. A formal definition is given below.

\begin{definition}
	\label{def.nucleolus}
	For an imputation $v: {V} \rightarrow \cR_+$, let $\theta(v)$ be the vector obtained by sorting the $2^{|V|} - 2$ values $v(S) - p(S)$ for each $\emptyset \subset S \subset V$ in non-decreasing order. Then the unique imputation, $v$, that lexicographically maximizes $\theta(v)$ is called the {\em nucleolus} and is denoted $\nu(G)$. 
\end{definition} 

In 1998, \cite{Faigle1998nucleon} stated the problem of computing the nucleolus of the matching game in polynomial time. For the unit weight case, this was done by Kern and Paulusma \cite{Kern2003matching}, and for arbitrary weights by Konemann et al. \cite{Konemann2020computing}. However, their algorithm makes extensive use of the ellipsoid algorithm and is therefore neither efficient nor does it give deep insights into the underlying combinatorial structure. They leave the open problem of finding a combinatorial polynomial time algorithm. We note that the difference $v(S) - p(S)$ appearing in the least core and nucleolus has not been upper-bounded for any standard family of games, including the general graph matching game. 

A different notion was recently proposed in \cite{Va.general}, namely {\em approximate core}, and was used to obtain an imputation in the $2/3$-approximate core for the general graph matching game. This imputation can be computed in polynomial time, again using the power of LP-duality theory. This result was extended to $b$-matching games in general graphs by Xiaoet al. \cite{b-matching-approximate}. 

Konemann et al. \cite{b-matching-Konemann} showed that computing the nucleolus of the constrained bipartite $b$-matching game, in which an edge can be matched at most once, is NP-hard even for the case $b=3$ for all vertices. Biro et al. \cite{Biro2012computing} showed that the core non-emptiness and core membership problems for the  $b$-matching game are solvable in polynomial time if $b \leq 2$ and are co-NP-hard even for $b = 3$.

Deng and Papadimitriou \cite{Papa-Deng1994complexity} define several computational questions regarding the core, e.g., testing non-emptyness, membership etc. and study them for a specific game on an undirected graph with positive and negative weight edges. Granot and Huberman \cite{Granot1981minimum, Granot-2-1984core} showed that the core of the minimum cost spanning tree game is non-empty and gave an algorithm for finding an imputation in it. Megiddo \cite{Megiddo1978computational} shows how to compute the nucleolus and Shapley value of this game. Koh and Sanita \cite{Laura-Sanita} settle the question of efficiently determining if a spanning tree game is submodular; the core of such games is always non-empty. Nagamochi et al. \cite{Nagamochi1997complexity} characterize non-emptyness of core for the minimum base game in a matroid; the minimum spanning tree game is a special case of this game. For the $b$-matching game, non-emptyness of the core was established by showing that every optimal dual solution is a core imputation \cite{Va.Char}; this involved exploiting total unimodularity of the constraint matrix. 




%% file: Vertices.tex
\section{The Core of the Assignment Game}
\label{sec.Complementarity}

In this section, we will give a short overview of the Shapley-Shubik work as well as the framework of Deng et al. \cite{Deng1999algorithms}, since they provide a valuable background for the rest of the paper.
 
The following setting, taken from \cite{Eriksson2001stable} and \cite{Biro2012computing}, vividly captures the issues underlying profit-sharing in an assignment game. Suppose a coed tennis club has sets $U$ and $V$ of women and men players, respectively, who can participate in an upcoming mixed doubles tournament. Assume $|U| = m$ and $|V| = n$, where $m, n$ are arbitrary. Let $G = (U, V, E)$ be a bipartite graph whose vertices are the women and men players and an edge $(i, j)$ represents the fact that agents $i \in U$ and $j \in V$ are eligible to participate as a mixed doubles team in the tournament. Let $w$ be an edge-weight function for $G$, where $w_{i j} > 0$ represents the expected earnings if $i$ and $j$ do participate as a team in the tournament. The total worth of the game is the weight of a maximum weight matching in $G$.

Assume that the club picks such a matching for the tournament. The question is how to distribute the total profit among the agents --- strong players, weak players and unmatched players --- so that no subset of players feel they will be better off seceding and forming their own tennis club. 

\begin{definition}
	\label{sec.coalition}
The {\em assignment game}, as defined above, will be denoted by $G = (U, V, E), \ w: E \rightarrow \cR_+$. The set of all players, $U \cup V$, is called the {\em grand coalition}. A subset of the players, $(S_u \cup S_v)$, with $S_u \subseteq U$ and $S_v \subseteq V$, is called a {\em coalition} or a {\em sub-coalition}.
\end{definition}

\begin{definition}
	\label{def.worth}
	The {\em worth} of a coalition $(S_u \cup S_v)$ is defined to be the maximum profit that can be generated by teams within $(S_u \cup S_v)$ and is denoted by $p(S_u \cup S_v)$. Formally, $p(S_u \cup S_v)$ is the weight of a maximum weight matching in the graph $G$ restricted to vertices in $(S_u \cup S_v)$ only. $p(U \cup V)$ is called the {\em worth of the game}. The {\em characteristic function} of the game is defined to be $p: 2^{U \cup V} \rightarrow \cR_+$.   
\end{definition}

\begin{definition}
	\label{def.imputation}	
	An {\em imputation}\footnote{Some authors prefer to call this a pre-imputation, while using the term imputation when individual rationality is also satisfied.} gives a way of dividing the worth of the game, $p(U \cup V)$, among the agents. It consists of two functions $u: {U} \rightarrow \cR_+$ and $v: {V} \rightarrow \cR_+$ such that $\sum_{i \in U} {u(i)} + \sum_{j \in V} {v(j)} = p(U \cup V)$. 
\end{definition}
	
\begin{definition}
	\label{def.core}
	An imputation $(u, v)$ is said to be in the {\em core of the assignment game} if for any coalition $(S_u \cup S_v)$, the total worth allocated to agents in the coalition is at least as large as the worth that they can generate by themselves, i.e., $\sum_{i \in S_u} {u(i)} +  \sum_{j \in S_v} {v(j)} \geq p(S)$.
\end{definition}
 
We next describe the characterization of the core of the assignment game given by Shapley and Shubik \cite{Shapley1971assignment}. As stated in Definition \ref{def.worth}, the worth of the game, $G = (U, V, E), \ w: E \rightarrow \cR_+$, is the weight of a maximum weight matching in $G$. Linear program (\ref{eq.core-primal-bipartite}) gives the LP-relaxation of the problem of finding such a matching. In this program, variable $x_{ij}$ indicates the extent to which edge $(i, j)$ is picked in the solution. Matching theory tells us that this LP always has an integral optimal solution \cite{LP.book}; the latter is a maximum weight matching in $G$.

	\begin{maxi}
		{} {\sum_{(i, j) \in E}  {w_{ij} x_{ij}}}
			{\label{eq.core-primal-bipartite}}
		{}
		\addConstraint{\sum_{(i, j) \in E} {x_{ij}}}{\leq 1 \quad}{\forall i \in U}
		\addConstraint{\sum_{(i, j) \in E} {x_{ij}}}{\leq 1 }{\forall j \in V}
		\addConstraint{x_{ij}}{\geq 0}{\forall (i, j) \in E}
	\end{maxi}

Taking $u_i$ and $v_j$ to be the dual variables for the first and second constraints of (\ref{eq.core-primal-bipartite}), we obtain the dual LP: 

 	\begin{mini}
		{} {\sum_{i \in U}  {u_{i}} + \sum_{j \in V} {v_j}} 
			{\label{eq.core-dual-bipartite}}
		{}
		\addConstraint{ u_i + v_j}{ \geq w_{ij} \quad }{\forall (i, j) \in E}
		\addConstraint{u_{i}}{\geq 0}{\forall i \in U}
		\addConstraint{v_{j}}{\geq 0}{\forall j \in V}
	\end{mini}

\begin{theorem}
	\label{thm.SS}
	(Shapley and Shubik \cite{Shapley1971assignment})
The imputation $(u, v)$ is in the core of the assignment game if and only if it is an optimal solution to the dual LP, (\ref{eq.core-dual-bipartite}). 
\end{theorem}

By Theorem \ref{thm.SS}, the core of the assignment game is a convex polyhedron. Finally, we state a fundamental fact about LP (\ref{eq.core-primal-bipartite}); its proof follows from the fact that the constraint matrix of this LP is totally unimodular. 

\begin{theorem}
	\label{thm.Birkhoff}
	(Birkhoff \cite{Birkhoff1946three})	The vertices of the polytope defined by the constraints of LP (\ref{eq.core-primal-bipartite}) are $0/1$ vectors, i.e., they are matchings in $G$.
\end{theorem}

\subsection{The Framework of Deng et al. \cite{Deng1999algorithms}}
\label{sec.Deng}

In this section, we present the framework of Deng et al. \cite{Deng1999algorithms}, which was mentioned in the Introduction. Let $T = \{1, \cdots, n\}$ be the set of $n$ agents of the game. Let $w \in \R^m_+$ be an $m$-dimensional non-negative real vector specifying the weights of certain objects; in the assignment game, the objects are edges of the underlying graph. Let $A$ be an $n \times m$ matrix with $0/1$ entries whose $i^{th}$ row corresponds to agent $i \in T$. Let $x$ be an $m$-dimensional vector of variables and $\mathbb{1}$ be the $n$-dimensional vector of all 1s. Assume that the worth of the game is given by the objective function value of following integer program. 

	\begin{maxi}
		{} {w \cdot x}
			{\label{eq.IP}}
		{}
		\addConstraint{}{Ax \leq \mathbb{1}}
		\addConstraint{}{x \in \{0, 1\}}
	\end{maxi}

Moreover, for a sub-coalition, $T' \subseteq T$ assume that its worth is given by the integer program obtained by replacing $A$ by $A'$ in (\ref{eq.IP}), where $A'$ picks the set of rows corresponding to agents in $T'$. The LP-relaxation of (\ref{eq.IP}) is:

	\begin{maxi}
		{} {w \cdot x}
			{\label{eq.Primal}}
		{}
		\addConstraint{}{Ax \leq \mathbb{1}}
		\addConstraint{}{x \geq 0}
	\end{maxi}

Deng et al. proved that if LP (\ref{eq.Primal}) always has an integral optimal solution, then the set of core imputations of this game is exactly the set of optimal solutions to the dual of LP (\ref{eq.Primal}). So far, for specific games, integrality of the vertices of the polytope defined by this linear system was established by showing that its constraint matrix is totally unimodular, as in Theorem \ref{thm.Birkhoff}.

%% file: Perfect.tex
\section{The Stable Set Game in Perfect Graphs}
\label{sec.perfect}

In Section \ref{sec.def-perfect} we give the required definitions and facts from the (extensive) theory of perfect graphs. We will not credit individual papers for these facts; instead, we refer the reader to Chapter 9 of the book \cite{GLS} as well as the remarkably clear and concise exposition of this theory, presented as an ``appetizer'' by Groetschel \cite{Grotschel1999my}.

\subsection{Definitions and Preliminaries}
\label{sec.def-perfect}

\begin{definition}
	\label{def.omega-chi}
	Given a graph $G = (V, E)$, $\omega(G)$ denotes its {\em clique number}, i.e., the size of the largest clique in it and $\chi(G)$ denotes its {\em chromatic number}, i.e., the minimum number of colors needed for its vertices so that the two endpoints of any edge get different colors. 
\end{definition}

\begin{definition}
	\label{def.perfect}
	A graph $G = (V, E)$ is said to be {\em perfect} if and only if for each vertex-induced subgraph $G' \subseteq G$, $\omega(G') = \chi(G')$.
\end{definition}

Let $\oG$ denote the {\em complement} of $G$, i.e., $\oG = (V, \oE)$, where $\oE$ is the complement of $E$, with $\forall \ u, v \in V, (u,v) \in E$ if and only if $(u, v) \notin \oE$. A central fact about perfect graphs is that $G$ is perfect if and only if $\oG$ is perfect. 

\begin{definition}
	\label{def.stable-set}
	$V' \subseteq V$ is said to be a {\em stable set} in $G$, also sometimes called an {\em independent set}, if no two vertices of $V'$ are connected by an edge, i.e., $\forall \ u, v \in V', \ (u, v) \notin E$. Let $w: V \rightarrow \QQ_+$ be a weight function on the vertices of $G$. Then the {\em worth} of the {\em stable set game} is defined to be the weight of a maximum weight stable set in $G$ and is denoted by $\worth(G)$.  
\end{definition}

The game defined above is well-motivated: $V$ represents a set of agents from which a subset needs to be chosen for performing a certain task. For agent $v \in V$, $w(v)$ specifies the contribution of agent $v$ to performing the task. An edge $(u, v)$ indicates that agents $u$ and $v$ do not get along, to the extent that a sub-coalition involving both of them is bound to fail. Therefore, the most effective sub-coalition is a maximum weight stable set in $G$. However, as is well known, the maximum stable set problem is NP-hard in arbitrary graphs, even with unit weights. For this reason, we restrict to perfect graphs where, as we will see below, it can be solved in polynomial time. 

\begin{remark}	
	\label{rem.complement}
	Since the complement of a perfect graph is also perfect, our result for the stable set game in perfect graphs yields an analogous result for the clique game in perfect graphs. The latter game is also well-motivated. 
\end{remark}

Let $G$ be an arbitrary graph. Clearly any clique in $G$ can intersect a stable set in at most one vertex, and therefore the constraint in LP (\ref{eq.stable-primal}) is satisfied by every stable set; note that variable $x_v$ indicates the extent to which $v$ is picked in a fractional stable set. LP (\ref{eq.stable-primal}) contains such a constraint for each clique in $G$ and is an LP-relaxation of the maximum weight stable set problem in $G$. However, LP (\ref{eq.stable-primal}) has exponentially many constraints, one corresponding to each clique in $G$; moreover, it is NP-hard to solve in general \cite{GLS}.

	\begin{maxi}
		{} {\sum_{v \in V}  {w_{v} x_{v}}}
			{\label{eq.stable-primal}}
		{}
		\addConstraint{x(Q)}{\leq 1 \quad}{\forall \ \mbox{clique $Q$ in} \ G}
		\addConstraint{x_{v}}{\geq 0}{\forall v \in V}
	\end{maxi}

The situation is salvaged in case $G$ is a perfect graph: via a rather long route, touching profound ideas, e.g., a foray into the Lovasz theta function and Shannon entropy, one can can show that LP (\ref{eq.stable-primal}) can be solved in polynomial time, see \cite{GLS, Grotschel1999my}. This requires the use of the ellipsoid algorithm; unfortunately, at present, an efficient combinatorial algorithm is not known. 

Another useful fact is that the constraint system of LP (\ref{eq.stable-primal}) is TDI for perfect graphs, see Definition \ref{def.TDI}. As a consequence, the polytope defined by this system has integral vertices, i.e., they are stable sets. This together with the previous fact implies that in perfect graphs, the worth of the stable set game can be computed in polynomial time. Furthermore, since for each vertex-induced subgraph $G' \subseteq G$, $G'$ is also perfect, this holds for the game restricted to $G'$ as well.

In order to characterize core imputations of this game, we will need the dual LP. This is obtain by taking $y_Q$ to be the dual variable for the constraint of LP (\ref{eq.stable-primal}), and is given in LP (\ref{eq.stable-dual}) below.

 	\begin{mini}
		{} {\sum_{\mbox{clique $Q$ in} \ G}  {y_Q}} 
			{\label{eq.stable-dual}}
		{}
		\addConstraint{\sum_{Q \ni v} {y_Q}}{ \geq w_{v} \quad }{\forall v \in V}
		\addConstraint{y_{Q}}{\geq 0}{\forall \ \mbox{clique $Q$ in} \ G}
	\end{mini}

The dual LP is solving a clique covering problem. In particular, if $w$ is 1 for every vertex, the minimum number of cliques which cover each vertex at least once is an optimal solution to the dual -- this follows from the fact that its constraint system is TDI.

\subsection{Characterizing the Core}
\label{sec.stable-core}

An {\em imputation} describes how the worth of the game leads to allocation of ``profits'' to sub-coalitions. In the assignment game, an optimal dual distributes the worth of the game among the  agents, i.e., the vertices of the graph. The allocation to a sub-coalition is simply the sum of the allocations to the agents in it. 

In our setting, an optimal dual distributes the worth of the game among the cliques of $G$. However, dividing the worth $y_Q$, given to clique $Q$, among the agents in the clique in not very meaningful. Instead, for a clique $Q$, we will view {\em the worth $y_Q$ as residing in every sub-clique of $Q$}. However, this is done not by simply duplicating the worth onto each sub-clique, an operation that is clearly not legitimate. It is done in a more interesting manner, so as to define allocations to sub-coalitions in a consistent manner.  

Before providing further details, let us establish terminology that is more appropriate for the setting at hand. We will think of $y_Q$ as the ``satisfaction'' of clique $Q$. Consider a sub-coalition $T \subseteq V$ such that $Q \cap T \neq \emptyset$, and let $Q' = Q \cap T$. As stated above, the satisfaction of $y_Q$ resides on every sub-clique of $Q$, including $Q'$. The crucial point is that {\em the latter viewpoint will manifest itself only when we ``think'' of the sub-coalition $T$}. Moreover, {\em $Q$ will allocate only to its largest sub-clique in $T$}. The details given below will clarify these ideas. 

Let $y:  (\mbox{cliques in} \ G) \rightarrow \RR_+$ be a function assigning satisfaction to the  cliques of $G$. Define
\[ \sat(G) = \sum_{\mbox{clique $Q$ in} \ G}  {y_Q} .\]

Let $T \subseteq V$ and let $G'$ be the subgraph of $G$ induced on $T$. Each clique $Q$ in $G$ will allocate satisfaction to at most one clique in $G'$ as follows. If $Q \cap T = \emptyset$, then $Q$ will make no allocation to any clique in $G'$. Otherwise, it will allocate $y_Q$ to $Q' = Q \cap T$. In general, there may be several cliques in $G$ which may allocate to $Q'$. Hence the total allocation to $Q'$ is defined to be
\[ z_{Q'} = \sum_{Q \ {\mbox{in} \ G}: Q \cap T = Q'} {y_Q} ,\]
where $z:  (\mbox{cliques in} \ G') \rightarrow \RR_+$ is the function which assigns satisfaction to the cliques of $G'$. 

Finally, define 
\[ \sat(G') = \sum_{\mbox{clique $Q'$ in} \ G'}  {z_{Q'}} .\]

\begin{definition}
	\label{def.stable-core}
	We will say that imputation $y$ is in the core of the stable set game if
	\begin{enumerate}
		\item  $ \worth(G) = \sat(G) $
		\item  $ \forall G' \subseteq G, \ \worth(G') \leq \sat(G')  $
	\end{enumerate}
\end{definition}
 
\begin{remark}
 	\label{rem.top-down}
The contrast between the way allocations are made to sub-coalitions in the assignment game and in our game can be summarized as follows: Whereas the former follows a bottom-up process, we follow a top-down process. To the best of our knowledge, all previous works characterizing the cores of games follow a bottom-up process, similar to the assignment game.  
 \end{remark}

\begin{theorem}
	\label{thm.stable}
	An imputation $y$ is in the core of the stable set game over a perfect graph if and only if it is an optimal solution to the dual, LP(\ref{eq.stable-dual}). 
\end{theorem}
	
\begin{proof}
	($\Leftarrow$) First assume that $y$ is an optimal solution to the dual, LP(\ref{eq.stable-dual}). By the LP-duality theorem and the TDI of its linear system, $\worth(G) = \sat(G)$. 
	
	Let $T \subseteq V$ and let $G'$ be the graph induced on $T$. By the allocation procedure described above, for $v \in T$, each clique $Q$ in $G$ containing $v$ will contribute $y_Q$ to its largest sub-clique $Q'$ in $G'$, which will contain $v$. Therefore, 
		\[ \forall v \in T: \ \sum_{Q' \ni v} {z_{Q'}} = \sum_{Q \ni v} {y_Q} \geq w_v .\]
	Therefore $z$ is feasible for the dual LP (\ref{eq.stable-dual}), restricted to $G'$. Hence by the weak duality theorem, we get that the objective value function of the primal LP restricted to $G'$ is bounded by $\sat(G')$. Finally, since this restricted system is also TDI, we get $\worth(G') \leq  \sat(G')$. 
	
	($\Rightarrow$) Next, assume that imputation $y$ is in the core of the stable set game over a perfect graph. By definition, $\worth(G) = \sat(G) = \sum_{\mbox{clique $Q$ in} \ G}  {y_Q} $.
	
	Let $v \in V$ and let $G'$ be $G$ restricted to $v$. By the definition of core and the way $y$ yields $z$, 
	\[ w_v = \worth(G') \leq \sat(G') = \sum_{Q' \ni v} {z_{Q'}} = \sum_{Q \ni v} {y_Q} .\]
	Therefore, $y$ satisfies the constraint for each $v \in V$ in the dual, LP(\ref{eq.stable-dual}). Hence $y$ is dual feasible. This together with the previous statement gives that $y$ is an optimal solution to the dual, LP(\ref{eq.stable-dual}). 
\end{proof}

\begin{corollary}
	\label{cor.stable}
	The core of the stable set game over a perfect graph is non-empty. 
\end{corollary}

%% file: Polymatroid.tex
\section{The Matroid Independent Set Game, Given via a Rank Oracle}
\label{sec.polymatroid}

In Section \ref{sec.def-polymatroid}, we present the necessary background information. Once again, we will not credit individual papers for these (mostly well-known) facts; instead we refer the reader to Chapter 10 in \cite{GLS}.

\subsection{Definitions and Preliminaries}
\label{sec.def-polymatroid}

Let $\cM = (U, \cI)$ be a matroid, where $U$ is the ground set with $n$ elements and $\cI$ is its set of independent sets. Let $w: U \rightarrow \RR_+$ be a function that assigns weights to the elements. Assume further that we are given a {\em rank oracle} $r$: for any set $S \subseteq U$, $r(S)$ is the rank of $S$, i.e., the size of a maximum independent set in $S$. The rank function of a matroid is {\em submodular}, i.e., 
\[ r(S) + r(T) \geq r(S \cap T) + r(S \cup T) .\] 

It is well known that a maximum weight independent set can be found via the following greedy algorithm. Sort the elements of $U$ by decreasing weight, say $e_1, e_2, \ldots , e_n$. For $1 \leq i \leq n$ do: Pick $e_i$ if $r(e_1, \ldots e_i) > r(e_1, \ldots, e_{i-1})$.

\begin{definition}
	\label{def.matroid-game}
	Given matroid $\cM = (U, \cI)$ and weight function $w$ on its elements, the worth of the {\em matroid independent set game} is defined to be the weight of a maximum weight independent set in $\cM$ and is denoted by $\worth(\cM)$. 
\end{definition}

We wish to characterize the core imputations of the matroid independent set game. LP (\ref{eq.matroid-primal}) is the LP-relaxation of the maximum weight independent set problem;   here $x_e$ is the extent to which element $e$ is picked. 

	\begin{maxi}
		{} {\sum_{e \in U}  {w_{e} x_{e}}}
			{\label{eq.matroid-primal}}
		{}
		\addConstraint{x(S)}{\leq r(S) \quad}{\forall S \subseteq U}
		\addConstraint{x_{e}}{\geq 0}{\forall e \in U}
	\end{maxi}

The linear system of LP (\ref{eq.matroid-primal}) is TDI and therefore its polytope has integral vertices, i.e., they are independent sets of the matroid. This polytope is called a {\em polymatroid}. Taking $u_i$ and $v_j$ to be the dual variables for the first and second constraints of (\ref{eq.matroid-primal}), we obtain the dual LP: 

 	\begin{mini}
		{} {\sum_{S \subseteq U}  {r(S) \cdot y_S}} 
			{\label{eq.matroid-dual}}
		{}
		\addConstraint{\sum_{S \ni e} {y_S}}{ \geq w_{e} \quad }{\forall e \in U}
		\addConstraint{y_{S}}{\geq 0}{\forall S \subseteq U}
	\end{mini}

Given an oracle for the rank function $r$, the dual LP can be solved in polynomial time using the ellipsoid algorithm \cite{GLS}.

\subsection{Characterizing the Core}
\label{sec.matroid-core}

The process of allocating the worth of the game to sub-coalitions is similar to that followed in the stable set game over perfect graphs, i.e., it is top-down, as described in Remark \ref{rem.top-down}. An optimal dual distributes the worth of the game among subsets $S \subseteq U$. Once again, dividing the worth $y_S$, given to set $S$, among the agents in the set in not very meaningful. Instead, again we will view the worth $y_S$ as residing in every subset of $S$, and this viewpoint manifests itself only when we ``think'' of the sub-coalition $T$. As before, we will think of $y_S$ as the ``satisfaction'' of set $S$. 

\begin{remark}
	\label{rem.confusion}
	An added source of confusion in this setting is that the worth of the game is distributed among subsets of $U$, and sub-coalitions are also subsets of $U$. We will try to mitigate this confusion by naming the former $S, S'$ etc. and the latter as $T, T'$ etc. 
\end{remark}

Let $y:  2^U \rightarrow \RR_+$ be a function assigning satisfaction to the subsets of $U$. Define
\[ \sat(U) = \sum_{S \subseteq U}  {y_S} .\]

Consider sub-coalition $T \subseteq U$. Each subset $S$ of $U$ will allocate satisfaction to at most one subset of $T$ as follows. If $S \cap T = \emptyset$, then $S$ will make no allocation to any subset of $T$. Otherwise, it will allocate $y_S$ to the set $S' = S \cap T$, i.e., the largest subset of $S$ in sub-coalition $T$. Once again, there may be several subsets of $U$ which may allocate to $S'$. The total allocation to $S'$ is defined to be
\[ z_{S'} = \sum_{S \subseteq U: S \cap T = S'} {y_S} ,\]
where $z: 2^T \rightarrow \RR_+$ is the function which assigns satisfaction to subsets of $T$. 

Finally, define 
\[ \sat(T) = \sum_{S' \subseteq T}  {z_{S'}} .\]

\begin{definition}
	\label{def.matroid-core}
	We will say that imputation $y$ is in the core of the matroid independent set game if
	\begin{enumerate}
		\item  $ \worth(U) = \sat(U) $
		\item  $ \forall T \subseteq U, \ \worth(T) \leq \sat(T)  $
	\end{enumerate}
\end{definition}

\begin{theorem}
	\label{thm.matroid}
	An imputation $y$ is in the core of the matroid independent set game if and only if it is an optimal solution to the dual, LP(\ref{eq.matroid-dual}). 
\end{theorem}
	
\begin{proof}
	($\Leftarrow$) First assume that $y$ is an optimal solution to the dual, LP(\ref{eq.matroid-dual}). By the LP-duality theorem and the TDI of this linear system, $\worth(U) = \sat(U)$. 
	
	Consider sub-coalition $T \subseteq U$. By the allocation procedure described above, 
		\[ \forall e \in T: \ \sum_{S' \subseteq T: \ S' \ni e} {z_{S'}} = \sum_{S \subseteq U: \ S \ni e} {y_S} \geq w_e .\]
	Therefore $z$ is feasible for the dual LP (\ref{eq.matroid-dual}), restricted to $T$. Therefore, by the weak duality theorem, we get that the objective value function of the primal LP restricted to $T$ is bounded by $\sat(T)$. Finally, since this restricted system is also TDI, we get $\worth(T) \leq  \sat(T)$. 
	
	($\Rightarrow$) Next, assume that imputation $y$ is in the core of the stable set game over a perfect graph. By definition, $\worth(U) = \sat(U) = \sum_{S \subseteq U} {y_S}$. 
	
	Let $e \in U$ and let $T'$ be the sub-coalition $\{e\}$. By the definition of core and the way $y$ yields $z$, 
	\[ w_e = \worth(T') \leq \sat(T') = \sum_{S' \subseteq T': \ S' \ni e} {z_{S'}} = \sum_{S \subseteq U: \ S \ni e} {y_S} .\]
	Therefore, $y$ satisfies the constraint for each $e \in U$ in the dual LP (\ref{eq.matroid-dual}). Hence $y$ is dual feasible. This together with the previous statement gives that $y$ is an optimal solution to the dual LP (\ref{eq.matroid-dual}). 
\end{proof}

\begin{corollary}
	\label{cor.matroid}
	The core of the matroid independent set game is non-empty. 
\end{corollary}

%% file: ack.tex
\section{Acknowledgements}
\label{sec.ack}

I wish to thank Gerard Cornuejols, Federico Echenique, Martin Groetschel, Herv\'e Moulin and Joseph Root for valuable discussions.